\documentclass[a4paper,UKenglish,cleveref, autoref, thm-restate]{lipics-v2021}

\hideLIPIcs  

\usepackage{listings}
\usepackage[mathlines]{lineno}
\usepackage{etoolbox}
\usepackage{url}
\usepackage[obeyFinal]{todonotes}
\usepackage[utf8]{inputenc}
\usepackage{csquotes}
\usepackage{rotating}
\usepackage{graphicx}
\usepackage{multirow}
\usepackage{multicol}
\usepackage{tikz}
\usetikzlibrary{arrows}
\usepackage{macros}

\usepackage{comment}

\nolinenumbers

\newif\ifshortversion

\ifshortversion
  \includecomment{shortonly}
  \excludecomment{longonly}
\else
  \includecomment{longonly}
  \excludecomment{shortonly}
\fi

\usepackage{amsthm}
\usepackage[capitalise]{cleveref}
\crefformat{footnote}{#2\footnotemark[#1]#3}

\title{Why not? Developing ABox Abduction beyond Repairs}

\author{Anselm Haak}{Knowledge Representation Group, Paderborn University, Germany}{anselm.haak@uni-paderborn.de}{https://orcid.org/0000-0003-1031-5922}{}

\author{Patrick Koopmann}{Knowledge in Artificial Intelligence, Vrije Universiteit Amsterdam, The Netherlands}{p.k.koopmann@vu.nl}{https://orcid.org/0000-0001-5999-2583}{}

\author{Yasir Mahmood}{Data Science Group, Paderborn University, Germany}{yasir.mahmood@uni-paderborn.de}{https://orcid.org/0000-0002-5651-5391}{}

\author{Anni-Yasmin Turhan}{Knowledge Representation Group, Paderborn University, Germany}{turhan@uni-paderborn.de}{https://orcid.org/0000-0001-6336-335X}{}

\authorrunning{A. Haak, P. Koopmann, Y. Mahmood, A.-Y. Turhan} 

\Copyright{Anselm Haak, Patrick Koopmann, Yasir Mahmood and Anni-Yasmin Turhan} 



\category{} 

\relatedversion{} 

\acknowledgements{
  We thank the anonymous reviewers for insightful comments and constructive feedback.
  The third author appreciates funding by DFG grant TRR 318/1 2021 – 438445824.
}

\EventEditors{John Q. Open and Joan R. Access}
\EventNoEds{2}
\EventLongTitle{42nd Conference on Very Important Topics (CVIT 2016)}
\EventShortTitle{CVIT 2016}
\EventAcronym{CVIT}
\EventYear{2016}
\EventDate{December 24--27, 2016}
\EventLocation{Little Whinging, United Kingdom}
\EventLogo{}
\SeriesVolume{42}
\ArticleNo{23}

\begin{document}

\maketitle

\begin{abstract}
  Abduction is the task of computing a sufficient extension	of a knowledge base (KB) that entails a conclusion not entailed by the original KB. 
	It serves to compute explanations, or \emph{hypotheses}, for such missing entailments.
	While this task has been intensively investigated	for perfect data and under classical semantics, less is known about abduction when erroneous data results in inconsistent KBs.
	In this paper we define a suitable notion of abduction under repair semantics, and propose a set of minimality criteria that guides abduction towards `useful' hypotheses.
	We provide initial  
	complexity results on deciding existence of and verifying abductive solutions with these criteria, under different repair semantics and for the description logics \DLLite and \ELbot.
\end{abstract}


\section{Introduction} 

In the context of description logic knowledge bases, the task of abduction is prominently used to explain missing consequences.
In general, given a theory and an \emph{observation}, that is a formula not entailed over the theory, abduction asks for a \emph{hypothesis}, which is a collection of statements to add to the theory in order to entail the observation.
For description logics, such hypotheses are often computed for a knowledge base and some kind of  Boolean query. This general task has been intensively investigated for description logics in many variants, depending on whether it is about extending the TBox~\cite{WeiKleinerDragisicLambrix2014,du2017practical,Haifani2022}, the ABox~\cite{Klarman2011,HallandBritzABox2012,Calvanese2013,Del-PintoS19,Ceylan2020,Koopmann21a,Homola2023}, both at the same time~\cite{Elsenbroich2006,Koopmann2020},
or operating on the level of concepts~\cite{bienvenu2008complexity,GlimmKW22}.
Prominent results range from complexity analysis~\cite{bienvenu2008complexity,Calvanese2013,WeiKleinerDragisicLambrix2014,Ceylan2020,Koopmann21a} to implemented systems~\cite{WeiKleinerDragisicLambrix2014,du2017practical,Koopmann2020,Haifani2022,Homola2023}, that are
sometimes integrated into user frontends~\cite{AlrabbaaBFHKKKK24,BoborovaKHP24}.

If abduction is applied to compute explanations, often  minimality criteria for the hypotheses are imposed to obtain \enquote{feasible} explanations. For example, it can be required that hypotheses are subset-minimal to facilitate small explanations \cite{Calvanese2013,WeiKleinerDragisicLambrix2014}.
Similarly, it can be of interest when generating explanations, to limit the hypotheses to a particular signature.
It has been shown that in this setting, referred to as \emph{signature-based abduction}, the complexity can be higher~\cite{Calvanese2013,Koopmann2020,Koopmann21a}.


In practical ontology-based applications, data is rarely free of errors and thus, the data that populates the ABox in the description logic KB can easily become inconsistent. In such cases,  everything would follow from the KB, but meaningful reasoning can be regained by resorting to some kind of inconsistency-tolerant, i.e.\ non-monotonic, semantics such as repair semantics \cite{BienvenuR13,LLRRS-JWS-15} or defeasible semantics \cite{CMMSV-ISWC-15,BFPS-AIJ-15,PT-JAR-18}. 
Repair semantics rely on restoring consistent versions of an inconsistent KB by removing minimal sets of conflicting ABox statements. Such a restored version is known as an \emph{ABox repair}. Depending on which, of the possibly many, repairs are considered for reasoning, different repair semantics have been defined and investigated in the literature mainly for ontology-mediated query answering (OMQA) settings, see \cite{Bienvenu-KIJ-20} for an overview.
Three fundamental repair semantics entail a Boolean query, if it holds w.r.t.\ some repair (\emph{brave semantics}), w.r.t.\ all repairs (\emph{AR semantics}) or w.r.t.\ the intersection of all repairs (\emph{IAR semantics}), respectively.

While explanations of query entailment under repair semantics have been investigated, explaining query non-entailment under these semantics has  been addressed to a much lesser extent. In particular, ABox abduction under repair semantics has not been studied thoroughly. 
In \cite{Calvanese2013} the explanation of negative query entailment  is
defined as an abductive task and investigated for \DLLite albeit under the classical semantics. The works on abduction under repair semantics  build on  their basic notions.
Abduction over inconsistent \DLLite KBs is studied in \cite{du2015towards} for IAR  semantics. They devise several minimality criteria and focus rather on computation algorithms for cases that are tractable w.r.t.\ data complexity.
In \cite{bienvenu2019computing}, the authors define explanations for positive and negative query entailment under  repair semantics. They investigate the data complexity of verifying (preferred) explanations for \DLLiteR and brave, AR and IAR semantics and show (in)tractability. We build on notions introduced in their paper and extend some of their results. 
A closely related setting is studied in  \cite{lukasiewicz2022explanations} for variants of Datalog$^{\pm}$. The authors concentrate on 
showing how  removal of facts in order to restore consistency, causes the non-entailment of the query and thus take a somewhat complementary view to  \cite{bienvenu2019computing}.

In this paper we study ABox abduction under repair semantics. We focus on flat ABox abduction, where the hypotheses use atomic concepts only and where the observation is a Boolean instance query (BIQ). 
%
%
We first need to adapt the basic definitions for abduction to the inconsistency-tolerant setting (in Section.~\ref{sec:abd}). Using repair semantics results in  some subtle differences  in comparison to abduction under classical semantics.  To address these, we make some conceptual contributions to adapt to the new setting. 
Since reasoning with the generated hypotheses is using repair semantics, we do not require the hypothesis itself to be consistent. This can lead to more ABox abduction results, obviously.
We extend the set of common minimality criteria for hypotheses to new ones that are dedicated to limit the (effect of) conflicts.

%
%
We show also some initial complexity results for two prominent decision problems  introduced for abduction \cite{Calvanese2013,Koopmann21a}. Given a KB and an observation,  the \emph{existence problem}, is to decide whether a hypothesis exists at all and the \emph{verification problem}, is to decide whether a given set of statements is a hypothesis. 
We examine  these problems for flat ABox abduction using observations that are atomic BIQs in regard of brave and AR semantics for the DLs \ELbot and \DLLite. Additionally, we cover the cases of preferred hypotheses that are subset-minimal or cardinality-minimal and also whether or not the signature is restricted.

It turns out (in Section~\ref{sec:ExPro}) that the existence problem considered without a signature restriction is trivial under brave semantics, but for AR semantics its complexity drops to that of the complement of brave entailment. 
Furthermore, deciding existence under signature restrictions keeps the same complexity of entailment for brave semantics, but for AR semantics it increases by one complexity level in the polynomial hierarchy for \ELbot.

The verification problem (treated in Section~\ref{sec:VePro}) does not become trivial for unrestricted signatures, but has the same complexity as entailment for general and $\leq$-minimal hypotheses. In case subset-minimality is required for hypotheses, we show that a more heterogeneous complexity  landscape unfolds. For instance, brave semantics incurs no or moderate increase in complexity depending on the DL.

\smallskip 
\noindent
\begin{shortonly}
All of the omitted proofs can be found in the long version of the paper that is appended to this submission.
\yasir{long version or technical report?}
\end{shortonly}

\section{Preliminaries} 
For a general introduction to description logics, we refer to the description logic textbook~\cite{DBLP:books/daglib/0041477}.
We assume familiarity with computational complexity~\cite{DBLP:books/daglib/0092426}, in particular with the complexity classes $\NL, \Ptime, \NP, \co\NP$ and $\SigmaP$. 
Additionally, 
$\DP$ is the class of decision problems representable as the intersection of a problem in $\NP$ and a problem in $\co\NP$.

\subsection{The Description Logics Considered: \ELbot and \DLLite}
The syntax of  \ELbot  concepts is given by
\[C \Coloneqq C \sqcap C \mid \exists r.C \mid A \mid \top \mid \bot,\]
where $r$ and $A$ range over all concept and role names, respectively.
\ELbot TBoxes  contain finitely many \emph{concept inclusions} 
$C \sqsubseteq D$ for \ELbot concepts $C$ and $D$.

We consider the $\DLLite$ dialects $\DLLiteR$ and \DLLiteCore.
In \DLLiteR (underlying the OWL 2 QL profile), TBoxes may contain \emph{concept inclusions} of the form $B \subsum C$ and \emph{role inclusions} of the form $Q \sqsubseteq S$, where $B, C, Q$ and $S$ are generated by the following grammar:
%
%
\[
B \Coloneqq A \mid \exists Q, \qquad~
C \Coloneqq B \mid \neg B, \qquad~
Q \Coloneqq R \mid R^-, \qquad~
S \Coloneqq Q \mid \neg Q,
\]
where $A$ and $R$ range over all concept and role names, respectively.
Then \DLLiteCore restricts \DLLiteR by disallowing role inclusions, so only concept inclusions of the above form are allowed.

We study instance queries (IQs), which consist of a (complex) concept and a variable: $C(x)$.
Boolean instance queries (BIQs) are IQs that use an individual name instead of a variable: $C(a)$.

For the rest of the paper, the general term \DLLite refers to either \DLLiteCore or \DLLiteR.
We do so, since all of our results apply to both DLs, as 
our proofs only use properties shared by both DLs:
(1) 
entailment of atomic BIQs is \NL-complete under \brave  and \coNP-complete under \ar semantics,
(2) for a TBox \calT, subset-minimal \calT-inconsistent ABoxes \calA are of size $2$, where \calA is \calT-inconsistent, if $\tup{\calT, \calA} \models \bot$, and
(3) for a TBox \calT and atomic BIQ $\alpha$, minimal \calT-supports of $\alpha$ are of size $1$, where a \calT-support of $\alpha$ is an ABox \calA with $\tup{\calT, \calA} \models \alpha$.

\subsection{Repair Semantics}

If a knowledge base is inconsistent, repair semantics can \enquote{restore} consistent versions and admit meaningful reasoning again.
As it is common, we consider ABox repairs.
We define these as well as two common kinds of repair semantics next.

Let $\calK = \tup{\calT, \calA}$ be an inconsistent knowledge base and \query be a Boolean (conjunctive) query.
A \emph{repair} of $\calK$ is a subset $\calR \subseteq \calA$ such that $\tup{\calT, \calR} \not\models \bot$ and there is no strict superset $\calR' \supset \calR$ with these properties.
The somewhat dual notion is a \emph{conflict} or \emph{conflict set} \conflict, which is a subset of the ABox that is \calT-inconsistent and subset-minimal with this property.
We denote by $\Conf(\calK)$ the set of conflicts of \calK.
We recall entailment under brave~\cite{BienvenuR13} and AR semantics~\cite{LLRRS-JWS-15}:
\begin{itemize}
	\item $\calK\models_{\brave} \query$ if and only if there exists some repair $\calR$ of $\calK$ such that $\tup{\calT,\calR} \models \query$.
	\item $\calK\models_{\ar} \query$ if and only if $\tup{\calT,\calR}\models \query$ for every repair $\calR$ of $\calK$.
\end{itemize}
The complexity of query entailment under repair semantics is well understood~\cite{BienvenuB16}.
Precisely, checking entailment of atomic BIQs under $\brave$ semantics is $\NL$-complete for $\DLLite$ and $\NP$-complete for $\EL_\bot$ in combined complexity, whereas under \ar semantics it is $\co\NP$-complete for both DLs.


\section{ABox Abduction for Inconsistent KBs} 
\label{sec:abd}

The central task of abduction is to compute abductive hypotheses. We define these for non-entailed BIQs  under repair semantics.
%
%
\begin{definition}\label{def:hypotheses}
	Let $\calK = \tup{\calT, \calA}$ be an inconsistent KB, $\alpha$ be an atomic BIQ (called an \emph{observation}) and $\calS \in \{\brave, \ar\}$ such that $\calK \not\models_\calS \alpha$.
	Then, 
	a pair
	$\tup{\calK, \alpha}$ is called an \emph{$\calS$-abduction problem}.
	A solution for such a problem, called \emph{\calS-hypothesis}, is an ABox \Hyp such that $\tup{\calT, \calA \cup \Hyp} \models_\calS \alpha$.
  An \calS-hypothesis \Hyp is called
	\begin{enumerate}
		\item \emph{flat}, if $\Hyp$ contains no complex concepts;
		\item \emph{over $\Sigma$}, if \Hyp uses only names from signature $\Sigma$, where $\Sigma$ is a set of  concept, role and individual names;
		\item \emph{conflict-confining}, if $\Conf(\tup{\calT, \calA \cup \Hyp}) = \Conf({\calK})$.
	\end{enumerate}
\end{definition}
Note that for an \calS-abduction problem $\tup{\calK, \alpha}$ we require that \calK is inconsistent and $\calK\not\models_\calS \alpha$.
So, we consider only the so-called promise problem, i.e.\ the problem restricted to these particular inputs.
The restriction aligns with the intuition that one asks for an $\calS$-hypothesis if it is already known that the knowledge base is inconsistent and the observation is not $\calS$-entailed in $\calK$.
In contrast, if we instead assume that \calK is consistent and $\alpha$ is not entailed by \calK under classical semantics, we obtain classical abduction problems.
In this case, we call an ABox \Hyp \emph{hypothesis for $\alpha$ under classical semantics}, if $\tup{\calT, \calA \cup \Hyp}\not\models \bot$ and $\tup{\calT, \calA \cup \Hyp} \models \alpha$.

While the first two properties of \calS-hypotheses from Definition~\ref{def:hypotheses} are standard for abduction, the last one adapts the idea of a hypothesis not introducing any inconsistencies to the setting, where the KB is already inconsistent to begin with.
It can equivalently be defined by requiring that $\tup{\calT, \calR \cup \Hyp} \not\models\bot $ for every repair $\calR$ of $\calK$.
Note that this property might not always be desired.
%
We consider the following reasoning problems for 
a given $\Sem$-abduction problem.
\pagebreak[2]    
\begin{definition}[Reasoning Problems]\label{def:problems}\ 
	Given an  \calS-abduction problem $\tup{\calK, \alpha}$.
	\begin{enumerate}[topsep=-\parskip+\lineskip]
		\item The \emph{existence problem} asks whether 
		$\tup{\calK, \alpha}$ has a solution;
		\item The \emph{verification problem} asks whether 
		a given ABox \calH is a hypothesis for 
		$\tup{\calK, \alpha}$.
	\end{enumerate}
\end{definition}



To obtain hypotheses that are  meaningful for explanation purposes, minimality criteria that yield \emph{preferred hypotheses} have been defined already for abduction under classical semantics. We restate some of them and extend this set of criteria to also treat conflicts.
\begin{definition}\label{def:minimality}
	Let $\calS \in \{\brave, \ar\}$, $\tup{\calK, \alpha}$ be an $\calS$-abduction problem, where $\calK = \tup{\calT, \calA}$, and let $\Hyp$ be an \calS-hypothesis for $\tup{\calK, \alpha}.$ 
	Considering ${\preceq}\in \{\subseteq,\leq\}$, \Hyp is called
	\begin{enumerate}[topsep=-\parskip+\lineskip]
		\item \emph{$\preceq$-minimal}, if there is no $\calS$-hypothesis $\Hyp'$ for $\tup{\calK, \alpha}$ such that $\Hyp' \prec \Hyp$;
		\item \emph{$\preceq_c$-minimal}, if there is no \calS-hypothesis $\Hyp'$ for $\tup{\calK, \alpha}$ such that $\Conf(\tup{\calT, \calA \cup \Hyp'}) \prec \Conf(\tup{\calT, \calA \cup \Hyp})$.
%
	\end{enumerate}
\end{definition}
We  use the term \emph{subset-minimal} for $\subseteq$-minimal  and \emph{cardinality-minimal} for  $\leq$-minimal.
For any (reasonable) combinations of repair semantics, the above properties, and minimality criteria, we consider the corresponding computational problems introduced in \cref{def:problems}.	

Under certain repair semantics, already standard reasoning tasks such as query answering can behave in unexpected ways.  This also holds true for abduction of $\subseteq$-minimal \ar-hypotheses,  due to reasoning being inherently non-monotonic in this case,
as the following  interesting effect illustrates.
More precisely, the set of \ar-hypotheses for a given \ar-abduction problem $\tup{\calK, \alpha}$ does not need to be convex with respect to the subset-relation.
We illustrate this  by a small example KB $\calK = \tup{\calT, \emptyset}$ and ABoxes $\calA_1 \subsetneq \calA_2 \subsetneq \calA_3$ such that $\tup{\calT, \calA_1} \models_\ar D(a)$ and $\tup{\calT, \calA_3} \models_\ar D(a)$, but $\tup{\calT, \calA_2} \not\models_\ar D(a)$.
This can be achieved by defining the TBox and the ABoxes as follows: 
\begin{align*}
	\calT &\coloneqq \{ B_1 \sqcap B_2 \sqsubseteq \bot,~ 
	C_1 \sqcap C_2 \sqsubseteq \bot,\quad 
	B_1 \sqcap C_1 \sqsubseteq D,~ 
	B_2 \sqcap C_1 \sqsubseteq D,\quad 
	E \sqsubseteq D\}, \\
	\calA_1 &\coloneqq \{B_1(a), B_2(a), C_1(a)\}, \qquad
	\calA_2 \coloneqq \calA_1 \cup \{C_2(a)\},  \qquad
	\calA_3 \coloneqq \calA_2 \cup \{E(a)\}
\end{align*}
%
This effect implies that $\subseteq$-minimality cannot be checked \emph{locally} by only considering subsets that remove one assertion at a time.
Instead, one seems to need a \emph{global} check for all subsets. 


In classical abduction, one further considers \emph{semantically minimal} hypotheses $\calH$, for which there exists no hypothesis $\Hyp'$ such that
$\tup{\calT, \calA \cup \Hyp} \models \Hyp'$, but $\tup{\calT, \calA \cup \Hyp'} \not\models\Hyp$.
We argue that while such a minimality criterion is natural for \ar-semantics, its meaning is unclear for $\brave$-hypotheses. 
For instance, what does semantic minimality tell  about two \brave-hypotheses entailing the observation, but  in possibly different repairs?
Further exploration of this minimality criterion is therefore left for future work.

%


\begin{table}
\centering
\begin{tabular}{llcccc}
	\toprule
	DLs & Semantics~ & \multicolumn{2}{c}{Existence} & \multicolumn{2}{c}{Verification}  \\ 
	& & general & signature & $\leq$-min & $\subseteq$-min \\
	\midrule 
	\multirow{2}{*}{$\DLLite$}  
	& \brave
	& Trivial &  \NL &  \NL & \NL\\
	& \ar
	&\NL &  in $\SigmaP$ & \coNP & \DP-hard, in $\PiP$ \\
	\multirow{2}{*}{$\ELbot$}
	& \brave 
	& Trivial &  \NP & \NP & \DP  \\
	& \ar 
	&\coNP & \SigmaP  &\coNP & \DP-hard, in $\PiP$ \\
		\bottomrule
\end{tabular}
\caption{Complexity overview for existence problem
	and for verification of hypothesis under subset-  and cardinality minimality. Unless noted otherwise all results are completeness results.}
\end{table}


\section{Existence Problem} 
\label{sec:ExPro}

We study in this section the complexity of the existence problem for both \ELbot and \DLLite, with and without a given signature, 
under brave and \ar semantics. 
Observe that for $\calS \in \{\brave, \ar\}$, the existence of any \calS-hypothesis implies the existence of a minimal one for all of the introduced minimality criteria.
Therefore, we only consider the existence problem for general $\calS$-hypotheses.
We begin with the case where no signature is given.
Further, the fact that the singleton set containing only the observation can be a hypothesis leads to the problem degenerating to a special case of entailment, or even becoming trivial.
This also lends additional motivation to study the signature-based setting next, where such trivial hypotheses can be prevented.

\subsection{Unrestricted Signature Hypothesis --- Admitting Trivial Hypotheses}
\label{sec:triv}

As we only consider atomic BIQs as observations $\alpha$, the set $\{\alpha\}$ is an ABox and, therefore, a candidate for a hypothesis for $\alpha$.
We study in this section how this trivial hypothesis affects the complexity of the existence problem for $\calS$-hypotheses, where $\calS\in\{\brave,\ar\}$.

Let $\tup{\calK, \alpha}$ be an \calS-abduction problem, where $\calK = \tup{\calT, \calA}$.
In case of $\calS = \brave$, it is easy to see that the set $\Hyp = \{\alpha\}$ is a \brave-hypothesis for $\tup{\calK, \alpha}$, as $\alpha$ is contained in some repairs of $\tup{\calT, \calA \cup \{\alpha\}}$.
Hence, a \brave-hypothesis always exists.
The case of AR semantics is slightly more interesting, as an \ar-hypothesis need not exist in general, even if trivial hypotheses are allowed.
Interestingly, in this case the complexity of the existence problem 
becomes a special case of \ar entailment that has the same complexity as non-entailment under brave semantics for both DLs.
In case of \ELbot, this means that checking existence of \ar-hypotheses has the same complexity as \ar entailment.
In contrast, for \DLLite this gives a complexity of $\co\NL = \NL$, which is the complexity of brave entailment.

\begin{longonly}
The remainder of this section is dedicated to the proof of the above finding in case of \ELbot and \DLLite, but the majority of findings holds independent of the DL under consideration.
The approach is as follows: We first show that, for a given KB \calK, there is any \ar-hypothesis for an assertion $\alpha$ if and only if $\{\alpha\}$ is an \ar-hypothesis for $\alpha$, and that this is equivalent to $\{\alpha\}$ being conflict-confining for \calK.
We then relate the problem of checking whether a singleton ABox is conflict-confining to the complement of a form of brave entailment and obtain the desired complexity result.

We begin by showing that, if the ABox $\{\alpha\}$ is conflict-confining for a given KB $\calK$, then it is an \ar-hypothesis for $\alpha$ in $\calK$.
We observe that in this case, $\{\alpha\}$ satisfies all minimality conditions from \cref{def:minimality}.
\end{longonly}

\begin{longonly}
\begin{lemma}\label{lem:conf-pres-1}
	Let $\tup{\calK, \alpha}$ be an \ar-abduction problem.
	If $\{\alpha\}$ is conflict-confining for $\calK$, then $\{\alpha\}$ is an \ar-hypothesis for $\alpha$ in $\calK$.
	In this case, $\{\alpha\}$ also is $\preceq$-minimal and $\preceq_c$-minimal for ${\preceq} \in \{\leq, \subseteq\}$.
\end{lemma}
\begin{proof}
	Let $\calK = \tup{\calT, \calA}$ and assume that $\{\alpha\}$ is conflict-confining for $\calK$.
	Let $\calK_\alpha \coloneqq \tup{\calT, \calA \cup \{\alpha\}}$.
	We show that $\calK_\alpha \models_{\text{AR}} \alpha$, that is, $\tup{\calT, \calR} \models \alpha$ for all repairs $\calR$ of $\calK_\alpha$.
  Consider any repair \calR of $\calK_\alpha$.
  It is sufficient to show that $\alpha \in R$, as this readily implies $\tup{\calT, \calR} \models \alpha$.
  For the sake of contradiction, assume $\alpha \not\in \calR$.
	As $\calR$ is a repair of $\calK_\alpha$, it is also a repair of $\calK$.
	As $\{\alpha\}$ is conflict-confining, we have $\tup{\calT, \calR \cup \{\alpha\}} \not\models \bot$.
	Therefore, $\calR$ is not maximally consistent, which is a contradiction.
	
  If $\{\alpha\}$ is conflict-confining, then it is $\preceq$-minimal by our assumption that $\calK\not\models_\ar\alpha$	and the fact that $\{\alpha\}$ is of size $1$.
  Furthermore, it is $\preceq_c$-minimal by definition as it does not introduce any additional conflicts.
\end{proof}
\end{longonly}

\begin{longonly}
Next, we show that $\{\alpha\}$ not being conflict-confining prevents \emph{any} \ar-hypotheses for $\alpha$.
\end{longonly}

\begin{longonly}
\begin{lemma}\label{lem:conf-pres-2}
	Let $\tup{\calK, \alpha}$ be an \ar-abduction problem.
	If $\{\alpha\}$ is not conflict-confining for $\calK$, then $\alpha$ has no \ar-hypothesis in $\calK$.
\end{lemma}
\begin{proof}
	Let $\calK = \tup{\calT, \calA}$.
  If $\{\alpha\}$ is not conflict-confining for $\calK$, there is a repair $\calR$ of $\calK$ such that $\tup{\calT, \calR \cup \{\alpha\}} \models \bot$.
	For the sake of contradiction, assume that there is a hypothesis $\Hyp$ of $\alpha$ in $\calK$, that is: for all repairs $\calR''$ of $\calK_{\Hyp} \coloneqq \tup{\calT, \calA \cup \Hyp}$, we have $\tup{\calT, \calR''} \models \alpha$.
	As $\calR$ is consistent with $\calT$, there is some repair $\calR'$ of $\calK_{\Hyp}$ with $\calR \subseteq \calR'$.
	But then we have $\tup{\calT, \calR'} \models \tup{\calT, \calR \cup \{\alpha\}} \models \bot$, which is a contradiction.
\end{proof}
\end{longonly}

\begin{longonly}
We now relate the problem of checking whether $\{\alpha\}$ is conflict-confining to non-entailment under brave semantics.
\end{longonly}

\begin{longonly}
\begin{lemma}\label{lem:conf-pres-redu}
	Let $\tup{\calK, \alpha}$ be an \ar-abduction problem.
	Then $\calK \not\models_{\text{brave}} \alpha$ if and only if $\{\neg\alpha\}$ is conflict-confining for $\calK$\cref{fn:negation}.
  Moreover, the same is true with the roles of $\alpha$ and $\neg\alpha$ swapped.
\end{lemma}
\begin{proof}
  Let $\calK = \tup{\calT, \calA}$.
	If $\calK \models_{\text{brave}} \alpha$, then there is some repair $\calR$ of $\calK$ with $\tup{\calT, \calR} \models \alpha$.
	But then $\tup{\calT, \calR \cup \{\neg \alpha\}} \models \bot$, so $\{\neg \alpha\}$ is not conflict-confining for $\calK$.
	
	On the other hand, if $\{\neg \alpha\}$ is not conflict-confining for $\calK$, then there is a repair $\calR$ of $\calK$ with $\tup{\calT, \calR \cup \{\neg \alpha\}} \models \bot$.
	As $\calR$ is consistent with $\calT$, but $\calR \cup \{\neg \alpha\}$ is not consistent with $\calT$, this implies that $\tup{\calT, \calR} \models \alpha$.
	This shows that $\calK \models_{\text{brave}} \alpha$.

  Note that the argument also applies when swapping the roles of $\alpha$ and $\neg\alpha$.
\end{proof}

The following Lemma combines the results of Lemmata~\ref{lem:conf-pres-1},~\ref{lem:conf-pres-2} and~\ref{lem:conf-pres-redu} to provide equivalent conditions for the existence of an \ar-hypothesis.
\end{longonly}

\begin{lemma}\label{lem:conf-pres}
  Let $\tup{\calK, \alpha}$ be an \ar-abduction problem.
  The following are equivalent:
  (1) there is an \ar-hypothesis for $\tup{\calK, \alpha}$, (2) $\{\alpha\}$ is an \ar-hypothesis for $\tup{\calK, \alpha}$, (3) $\{\alpha\}$ is conflict-confining for \calK, and (4) $\calK \not\models_\brave \neg\alpha$.
\begin{shortonly}
  \footnote{Here, entailment of $\neg \alpha$ has the usual meaning, even if negation is not in the logical language.}
\end{shortonly}
\begin{longonly}
  \footnote{\label{fn:negation} Here, $\{\neg\alpha\}$ being conflict-confining and entailment of $\neg\alpha$ have the usual meaning, even if negation is not in the logical language.}
\end{longonly}
\end{lemma}

\begin{longonly}
Finally, we use the previous Lemmata to show that the complexity of the existence problem for \ar-abduction problems degenerates to that of the complement of instance checking under brave semantics, both for \ELbot and \DLLite.

\end{longonly}

\begin{theorem}\label{thm:triv-complexity}
	The existence problem for \ar-hypotheses is \coNP-complete for \ELbot and \NL-complete for \DLLite.
  Moreover, the problem is trivial for \brave-hypotheses in both DLs.
\end{theorem}
\begin{longonly}
\begin{proof}
  The case of \brave-hypotheses directly follows from the fact, that $\{\alpha\}$ is a \brave-hypothesis for $\alpha$ in \calK for any atomic BIQ $\alpha$ and KB \calK.
  The remainder of the proof handles the case of \ar-hypotheses.
	We begin with the case of \ELbot.
  Let a KB $\calK = \tup{\calT, \calA}$ and observation $\alpha$ be given.
	
	\emph{Membership:} By \cref{lem:conf-pres}, checking whether there is an \ar-hypothesis for $\alpha$ in \calK is equivalent to checking whether $\tup{\calT, \calA \cup \{\alpha\}} \models_{\text{AR}} \alpha$.
	The latter can be checked in \co\NP, as it is a special case of instance checking under AR semantics.
	
	\smallskip
	\noindent
	\emph{Hardness:}
  By \cref{lem:conf-pres}, checking whether there is an \ar-hypothesis for $\alpha$ in \calK is equivalent to checking whether $\{\alpha\}$ is conflict-confining for \calK.
  Further, \cref{lem:conf-pres-redu} gives rise to a reduction from checking whether $\calK \not\models_{\text{brave}} \alpha$ to the problem of checking whether $\neg \alpha$ is conflict-confining for $\calK$.
	The latter is equivalent to checking whether $\{A'(a)\}$ is conflict-confining for $\tup{\calT', \calA}$, where $\calT' \coloneqq T \cup \{A \sqcap A' \sqsubseteq \bot\}$ for a fresh concept name $A'$.
	For this, notice that the set of repairs for $\calK$ and $\tup{\calT', \calA}$ are the same, as $\calA$ does not contain any axioms involving $A'$.
	It is now easy to see that for all repairs $\calR$, we have
	\begin{align*}
		\tup{\calT, \calR \cup \{\neg A(a)\}} \models \bot &\iff \tup{\calT, \calR} \models A(a) \\
		&\iff \tup{\calT', \calR} \models A(a) \\
		&\iff \tup{\calT', \calR \cup \{A'(a)\}} \models \bot.
	\end{align*}
	As instance checking in $\ELbot$ under brave semantics is \NP-hard, this implies \co\NP-hardness of the existence problem for \ar-hypotheses in \ELbot.
	
	\smallskip
	We now turn to \DLLite.
	\NL-hardness can be shown in  the same way as \coNP-hardness for \ELbot above, only changing the syntax of the disjointness axiom to $A \sqsubseteq \neg A'$.
	This uses the facts that instance checking in \DLLite under brave semantics is \NL-hard and \NL is closed under complement.
	As instance checking in \DLLite under AR semantics is \coNP-hard, we argue membership slightly differently from above.
	By~\cref{lem:conf-pres}, $\{\alpha\}$ is conflict-confining for 
	some $\calK = \tup{\calT, \calA}$ if and only if $\tup{\calT, \calA \cup \{\alpha\}} \models_\ar \alpha$.
	The latter is the case if and only if $\alpha$ is not contained in any conflict of $\tup{\calT, \calA \cup \{\alpha\}}$:
  The implication from right to left is obvious.
  On the other hand, if $\tup{\calT, \calA \cup \{\alpha\}} \models_\ar \alpha$, then for all repairs $\calR$ of $\tup{\calT, \calA \cup \{\alpha\}}$, $\calR \cup \{\alpha\}$ is \calT-consistent.
  Hence, $\alpha$ is contained in all repairs and cannot be contained in any conflict of $\tup{\calT, \calA \cup \{\alpha\}}$.
	As for \DLLite conflicts are always of size $2$, we can check whether $\alpha$ is contained in any conflict of $\tup{\calT, \calA \cup \{\alpha\}}$ by iterating over all assertions in $\calA$, and for each of them checking whether they become $\calT$-inconsistent together with $\alpha$.
	The latter can be done in $\NL$ for \DLLite.	
\end{proof}
\end{longonly}

Note that the equivalence of $\{\alpha\}$ being an \ar-hypothesis for $\alpha$ and $\{\alpha\}$ being conflict-confining means that this result applies both to general and conflict-confining \ar-hypotheses.
The set $\{\alpha\}$ being a conflict-confining \ar-hypothesis also implies that there is a conflict-confining \brave-hypothesis (namely $\{\alpha\}$).
Still, the complexity of the existence problem for conflict-confining \brave-hypothesis remains open:
There are cases where there is a conflict-confining \brave-hypothesis for $\alpha$, but $\{\alpha\}$ is not conflict-confining.

\subsection{Signature-based Setting --- Restricting the Signature of Hypotheses }

As we just have seen, checking existence of hypotheses without additional restrictions degenerates to entailment, or even a special case of entailment, because the observation itself can be a hypothesis.
A natural way to prevent this is to restrict the signature of hypotheses, that is, only consider hypothesis over some signature $\Sigma$ as defined in \cref{def:hypotheses}.

%
%

We begin by characterizing the complexity for consistent KBs under classical semantics. 
It turns out that this classical abduction problem is $\NP$-complete.
Then we consider the setting with inconsistent KBs under repair semantics and prove that the \NP-membership still holds under brave semantics. 
However, the complexity rises to $\SigmaP$-complete under \ar semantics.


\begin{theorem}\label{thm:consistent-signature-elbot}
  For \ELbot, the existence problem for hypotheses under classical semantics over a given signature $\Sigma$ is \NP-complete.
\end{theorem}
\begin{longonly}
\begin{proof}
	\emph{Membership:} Guess a set of assertions over the signature $\Sigma$ (i.e., guess $\Hyp \subseteq \{A(x), r(x,y) \mid A, r, x, y \in \Sigma\}$), then verify that $\calA \cup \Hyp$ is consistent with \calT and $\tup{\calT, \calA \cup \Hyp} \models \alpha$. 
	Both checks can be performed in polynomial time.
	
	\emph{Hardness:} We reduce from propositional satisfiability.
	To this aim, let $\varphi = \{c_1,\dots,c_p\}$ be a CNF formula over propositions $X =\{x_1,\dots,x_n\}$, where each $c_i$ is a clause.
	A literals $\ell$ is a variable $x$ or a negated variable $\neg x$.
	For a literal $\ell$, we denote by $\bar\ell$ its ``opposite'' literal.
	For a set $Z$ of variables, $\Lit(Z)$ denotes the collection of literals over $Z$.
	We construct the KB $\calK = \tup{\calT, \calA}$ with concept names $N=\{A_x, A_{\bar x} \mid x\in X\} \cup \{A_c \mid c \in\varphi\} \cup \{A_\varphi\}$, where
	\[
		\calT = \{A_x \sqcap A_{\bar x} \subsum \bot \mid x\in X\} \cup \{A_\ell \subsum A_c \mid \ell \in c, c\in\varphi\} \cup \{\sqcap_{c\in \varphi}A_c\subsum A_\varphi \},
	\]
	and $\calA = \emptyset$. 
	Further, let $\alpha \coloneqq A_\varphi(m)$ and $\Sigma =\{A_x, A_{\bar x}, \mid x\in X\} \cup \{m\}$ for an individual name $m$.
	Now, $\tup{\calK, \alpha}$ together with the signature $\Sigma$ is the desired abduction problem.
	Clearly, $\calK\not\models \alpha$ (since no axiom in $\calK$ involves $m$).
	\begin{claim}
		$\varphi$ is satisfiable if and only if $\alpha$ has a hypothesis over $\Sigma$ in \calK.
	\end{claim}
	\begin{claimproof}
  ``$\Longrightarrow$'': Let $s \subseteq \Lit(X)$ be a satisfying assignment for $\varphi$ seen as a set of literals.
	Then, for each clause $c\in\varphi$ there is some $\ell\in s\cap c$.
	We define $\Hyp = \{A_\ell(m) \mid \ell\in s\}$. 
	Since $s$ is an assignment, the set $\Hyp$ is consistent with $\calT$:
	No inconsistency is triggered due to axioms $A_x\sqcap A_{\bar x}\subsum \bot$, as $\Hyp$ only contains exactly one assertion for each variable.
	Now, we prove that $\calK_{\Hyp} \models \alpha$, where $\calK_{\Hyp} \coloneqq \tup{\calT, \Hyp}$.
	Since each clause is satisfied, we have $\calK_{\Hyp} \models A_c(m)$ for each $c\in\varphi$, which in turn implies that $\calK_{\Hyp} \models A_\varphi(m)$ due to the last TBox axiom.
	
	``$\Longleftarrow$'':
	Let $\Hyp$ be a hypothesis for $\alpha$ in $\calK$.
	Observe that there are sets $X1,X_2\subseteq X$ of variables such that $X_1 \cap X_2 = \emptyset$ and $\Hyp$ takes the following form:
	\[\Hyp = \{A_x(m) \mid x\in X_1\}\cup \{A_{\bar x}(m)\mid x\in X_2\}.\]
	This holds, since some disjointness axiom is violated otherwise.
	We define $s_{\Hyp} = \{\ell \mid A_\ell(m) \in \Hyp \}$.
	Then, $s_{\Hyp}$ is a potentially partial assignment over $X$, as for any variable it may not contain both the positive and the negative literal.
	Now, we prove that $s_{\Hyp} \models \varphi$.
	However, this is easy to see, since $\tup{\calT, \Hyp} \models A_c(m)$ for every clause $c \in \varphi$. 
	Consequently, for each $c \in \varphi$, there is some $\ell \in c$, such that $A_\ell(m) \in \Hyp$, which in turn implies that $\ell \in s_{\Hyp}$.
  Obviously, $s_\Hyp$ can also be extended to a full assignment that still satisfies $\varphi$.
	
	This concludes the correctness proof and establishes the claim.
  \end{claimproof}
\end{proof}
\end{longonly}

\paragraph*{The Inconsistent Case.}
Now we analyse the case of inconsistent KBs and repair semantics.

\begin{theorem}\label{thm:signature-complexity}
	For $\ELbot$, 
	the existence problem for $\calS$-hypotheses 
	over a given signature $\Sigma$ is 
	(1) $\NP$-complete for $\calS=\brave$, and
	(2) $\SigmaP$-complete for $\calS=\ar$.
\end{theorem}

\begin{proof}
	For (1): An \NP-algorithm for the problem can guess a hypothesis $\Hyp$ over the signature $\Sigma$ and, at the same time, guess a repair $\calR$ of the ABox.
  Then, verify that $\tup{\calT, \calR \cup \Hyp} \not\models \bot$ and $\tup{\calT, \calR \cup \Hyp} \models \alpha$ in polynomial time.
	The \NP-hardness can be shown by slightly adapting the reduction in Theorem~\ref{thm:consistent-signature-elbot}, adding an artificial inconsistency over fresh concepts not in $\Sigma$.
	
	For (2): The following algorithm shows \SigmaP-membership:
  Guess a set $\Hyp$ such that for all repairs $\calR$ of $\tup{\calT, \calA \cup \Hyp}$, we have $\tup{\calT, \calR} \models \alpha$.
	This requires $\NP$-time to guess the set $\Hyp$ and an $\NP$-oracle to guess a repair $\calR$ as a counter example to the entailment, thus resulting in $\SigmaP$-membership.
	For hardness, we reduce from the standard $\SigmaP$-complete problem $\QBF_2$:
	Instances of $\QBF_2$ are of the form $\Phi \coloneqq \exists Y \forall Z \varphi'$, where $\varphi'$ is a Boolean formula over variables $X = Y \cup Z$.
	Without loss of generality, we can assume that $\varphi' = \neg \varphi$ for some Boolean formula $\varphi$ in CNF.
	The problem asks whether $\Phi$ is valid (or true).
	We construct the following KB $\calK = \tup{\calT, \calA}$, using concept names $N = \{A_x, A_{\bar x} \mid x \in X\} \cup \{V_y \mid y \in Y\} \cup \{A_c \mid c \in \varphi\} \cup \{A_{\varphi}, A_{\bar\varphi}, C\}$.
	The TBox $\calT$ contains the following sets of axioms:
  \begin{align*}
 		\{A_x \sqcap A_{\bar x} \subsum \bot \mid x\in X\} &\quad \text{(ensures a valid assignment over $X$)}, \\
		\{A_\ell \subsum A_c \mid \ell \in c, c\in\varphi\} &\quad \text{(each clause is satisfied)}, \\
		\{\sqcap_{c\in \varphi}A_c\subsum A_\varphi , A_\varphi \sqcap A_{\bar\varphi}\subsum \bot\} &\quad \text{(the formula $\varphi$ is satisfied)}, \\
		\{A_y \subsum V_y, A_{\bar y} \subsum V_y \mid y\in Y\} &\quad \text{(hypotheses over $\Sigma$ are assignments over $Y$), and} \\
		\{\sqcap_{y\in Y} V_y \sqcap A_{\bar \varphi} \subsum C\} &\quad \text{(confirm the above with a concept name $C$)}. 
  \end{align*}

	Further, let $\calA \coloneqq \{A_z(m), A_{\bar z}(m)\mid z\in Z\}\cup \{A_{\bar\varphi}(m)\}$ for an individual name $m$.
	Finally, let $\Sigma \coloneqq \{m\} \cup \{A_y, A_{\bar y} \mid y \in Y\}$ and $\alpha \coloneqq C(m)$.
	Now $\tup{\calK, \alpha}$ together with the signature $\Sigma$ is the desired abduction problem.

  We first observe that $\tup{\calK, \alpha}$ is a valid \ar-abduction problem:
  Obviously, $\calK \models \bot$ when $Z$ is non-empty, due to both $A_z(m)$ and $A_{\bar z}(m)$ being present in the ABox for every $z \in Z$.
  Also, $\calK\not\models_{\ar} \alpha$, as $\calA$ does neither involve the concept name $C$ nor any of the concept names $A_y$, $A_{\bar y}$, or $V_y$ for $y \in Y$.
  The following claim states correctness of the reduction.
	\begin{claim}
		$\Phi$ is true if and only if $\alpha$ has an AR-hypothesis over the signature $\Sigma$ in $\calK$.
	\end{claim}
	
	\begin{claimproof}
	``$\Longrightarrow$'':
	Suppose $\Phi$ is true.
	Then there is an assignment $s \subseteq \Lit(Y)$ such that for all assignments $t \subseteq \Lit(Z)$, $\neg\varphi[s,t]$ is true.
\begin{shortonly}
  Here, $\Lit(\cdot)$ denotes the set of literals over a given set of variables.
\end{shortonly}
	We construct an AR-hypothesis for $\alpha$ from $s$.
	Let $\Hyp_s = \{A_p(m) \mid p\in s\}$.
	Obviously, $\Hyp_s$ is an ABox over $\Sigma$.
	Also, it does not violate any axiom of the form $A_y \sqcap A_{\bar y} \subsum \bot$, since $s$ is an assignment.
	
	We prove that $\tup{\calT, \calA \cup \Hyp_s} \models_{\ar} \alpha$.
	Consider any repair $\calR$ of $\tup{\calT, \calA \cup \Hyp_s}$.
	As $\tup{\calT, \calR} \not\models \bot$, \calR does not violate any axiom of the form $A_x \sqcap A_{\bar x} \sqsubseteq \bot$.
	Hence, $\calR \cap \{A_x(m), A_{\bar x}(m) \mid x \in X\}$ corresponds to (potentially partial) assignments $s_\calR \subseteq s$ and $t_\calR$ over $Y$ and $Z$, respectively.
	We first prove that $\tup{\calT, \calR} \not\models A_\varphi(m)$.
	Suppose to the contrary, that $\tup{\calT, \calR} \models A_\varphi(m)$.
	As \calR is consistent with \calT, this only happens by triggering the axiom $\sqcap_{c \in \varphi} A_c \sqsubseteq A_\varphi$, and in turn an axiom of the form $A_\ell \sqsubseteq A_c$ for each clause $c \in \varphi$.
	But this means that $s_\calR \cup t_\calR$, and hence also $s \cup t_\calR$, is a satisfying assignment for $\varphi$, which is a contradiction to $\neg\varphi[s,t]$ being true for all assignments $t$ over $Z$.
  Indeed, as this argument covers the case $s_\calR = s$, subset-maximality of repairs further yields that $\Hyp_s \subseteq \calR$.
	Moreover, subset-maximality together with the fact that $\tup{\calT, \calR} \not\models A_\varphi(m)$ yields that $A_{\bar \varphi}(m) \in \calR$.
	Consequently, $\tup{\calT, \calR} \models C(m)$.
	
	``$\Longleftarrow$'':
	Suppose $\Phi$ is false.
	Then, for each assignment $s \subseteq \Lit(Y)$, there is an assignment $t \subseteq \Lit(Z)$ such that $\neg\varphi[s,t]$ is false or, equivalently, $\varphi[s,t]$ is true.
	The latter can be stated as: each clause $c\in \varphi$ contains some literal $\ell \in c$ with $\ell\in s\cup t$.
	
	We now prove that $\alpha$ does not have any AR-hypothesis over $\Sigma$ in $\calK$.
	For 
	contradiction, assume that $\Hyp \subseteq \{A_p(m) \mid p \in \Lit(Y)\}$ is such a hypothesis and consider any repair \calR of $\tup{\calT, \calA  \cup \Hyp}$.
	As $\calR$ does not violate any axiom of the form $A_x \sqcap A_{\bar x} \sqsubseteq \bot$, the subset $\calR_Y \coloneqq \calR \cap \{A_y(m), A_{\bar y}(m) \mid y \in Y\}$ corresponds to a potentially partial assignment $s_\calR$ over $Y$.
	On the other hand, as $\tup{\calT, \calR} \models C(m)$, we also have $\tup{\calT, \calR} \models \sqcap_{y \in Y} V_y(m)$.
	Therefore, $\calR$ contains at least one assertion from $\{A_y(m), A_{\bar y}(m)\}$ for each $y\in Y$, i.e.\ that $s_\calR$ is a full assignment over $Y$.
	By our assumption, there is an assignment $t$ over $Z$ s.t.\ $\varphi[s_\calR, t]$ is true.
	
	Let $\calR_t \coloneqq \calR_Y \cup \{A_\ell(m) \mid \ell \in t\}$.
	Obviously, $\calR_t$ does not violate any of the disjointness axioms in \calT, as it does not contain $A_{\bar \varphi}(m)$ and $s_\calR \cup t$ is an assignment over $X$.
	This further means that $\tup{\calT, \calR_t} \not\models C(m)$.
	Furthermore, $\calR_t$ is subset-maximal: As both $s_\calR$ and $t$ are full assignments, we cannot add any assertion of the form $A_x(m)$ or $A_{\bar x}(m)$ for $x \in X$ without violating one of the disjointness axioms.
	Also, as $\varphi[s_\calR, t]$ is true, we have $\tup{\calT, \calR_t} \models A_\varphi(m)$.
	Hence, we also cannot add $A_{\bar \varphi}(m)$ without violating the corresponding disjointness axiom.
	This shows that $\calR_t$ is a repair of $\tup{\calT, \calA \cup \Hyp}$ that does not entail $\alpha$, contradicting our assumption.
	
	This proves the correctness of the claim and establishes the theorem.%
	\end{claimproof}
\end{proof}

We now turn to \DLLite, where we show that checking existence for \brave-hypotheses has the same complexity as \brave entailment.

\begin{theorem}\label{thm:signature-complexity-dllite}
	For $\DLLite$, the existence problem for $\brave$-hypotheses over a given signature $\Sigma$ is $\NL$-complete.
\end{theorem}
\begin{longonly}
\begin{proof}
  \emph{Membership:}
  Let $\Hyp_m$ be the set of all assertions over $\Sigma$.
  This set can be constructed in logarithmic space.
  It is easy to see that there is a \brave-hypothesis over $\Sigma$ for $\alpha$ in \calK if and only if $\Hyp_m$ is such a hypothesis.
  Hence, we only have to check whether $\tup{\calT, \calA \cup \Hyp_m} \models_\brave \alpha$, and membership follows as $\brave$-entailment for $\DLLite$ is in $\NL$.
	
  Hardness can be shown by a straightforward reduction from reachability in directed graphs.
	Let $G = (V,E)$ be a directed graph and $s,t \in V$.
	Define $\calT' \coloneqq \{A_{v_1} \sqsubseteq A_{v_2} \mid (v_1, v_2) \in E\}$.
	Now there is an $s$-$t$-path in $G$ if and only if $\tup{\calT', \{A_s(a)\}} \models A_t(a)$.
	To obtain a \brave-abduction problem, we add an artifical inconsistency.
	Let $\calK \coloneqq \tup{\calT, \calA}$, where $\calT \coloneqq \calT' \cup \{B_1 \sqsubseteq \neg B_2\}$ and $\calA \coloneqq \{B_1(b), B_2(b)\}$.
	Obviously, $\calK \models \bot$ and $\calK \not\models_\brave A_t(a)$.
	Furthermore, defining $\Sigma \coloneqq \{A_s, a\}$ it is easy to see that
  \begin{align*}
    \text{there is an $s$-$t$ path in $G$} &\Longleftrightarrow A_s(a) \text{ is a \brave-hypothesis for } A_t(a) \text{ in } \calK \\
    &\Longleftrightarrow \text{there is a \brave-hypothesis for } A_t(a) \text{ over } \Sigma \text{ in } \calK.
  \end{align*}
    As $\calK$ and $A_t(a)$ can be constructed from $G$ in logarithmic space, this shows \NL-hardness under logspace many-one reductions.
\end{proof}
\end{longonly}

Regarding $\ar$-semantics for $\DLLite$, it is easy to see that $\SigmaP$-membership can be shown in the same way as for \ELbot in the proof of Theorem~\ref{thm:signature-complexity}.
Determining the precise complexity for this case remains open for now.


\section{Verification Problem}
\label{sec:VePro}


The verification problem does not become quite as easy even without restricting the signature, so even if trivial hypotheses are allowed.
Interestingly, we even show that in case of $\subseteq$-minimality the complexity goes beyond that of entailment under repair semantics in some cases.
We begin with the case of general and $\leq$-minimal hypotheses for \ELbot, where the complexity of the corresponding entailment problem is inherited.



\begin{lemma}\label{lem:verification-el}
	For $\ELbot$, the verification problem for $\calS$-hypotheses is 
	(1) $\NP$-complete for $\calS=\brave$, and
	(2) $\co\NP$-complete for $\calS=\ar$. 
	This 
	also applies to $\leq$-minimal hypotheses.
\end{lemma}
\begin{longonly}
\begin{proof}
  Let $\tup{\calK, \alpha}$ be an \calS-abduction problem for \ELbot and \Hyp an ABox.
	Observe that the question whether $\Hyp$ is a $\calS$-hypothesis for $\tup{\calK, \alpha}$ is in fact $\calS$-entailment.
	Thus, $\calS$-verification is at most as hard as $\calS$-entailment.
	Furthermore, to check $\leq$-minimality it is sufficient to check whether $|\Hyp| = 1$, since $\{\alpha\}$ is an \calS-hypothesis for $\alpha$ in \calK if and only if there is any such hypothesis by \cref{lem:conf-pres}. 
	Thus, it suffices to determine whether $\tup{\calT,\calA\cup\Hyp} \models_\calS \alpha$.

	For hardness, we reduce $\calS$-entailment to $\calS$-verification.
	Given a KB $\calK = \tup{\calT,\calA}$ and a Boolean instance query $C(a)$, we let $\calT' = \calT \cup \{C\sqcap B \subsum A\}$ and consider $\calK'=\tup{\calT',\calA}$.
	Moreover, we let $\alpha \dfn A(a)$ and $\Hyp\dfn\{B(a)\}$.
	It is easy to see that $\calK' \not\models_\calS \alpha$ and 
  \[\calK \models_\calS C(a)\iff \calK'\models_\calS C(a) \iff \tup{\calT',\calA\cup\calH}\models_\calS \alpha.\]
	Consequently, $\tup{\calK, \alpha}$ is an \calS-abduction problem and $\calH$ is a ($\leq$-minimal) $\calS$-hypothesis for $\tup{\calK',\alpha}$ if and only if $\calK \models_\calS C(a)$. 	
\end{proof}
\end{longonly}


We prove next that the complexity of verification rises to $\DP$-completeness for $\subseteq$-minimal hypotheses.
The complexity gap between verifying $\subseteq$ and $\leq$ hypotheses seems somewhat surprising at first.
Nevertheless, the \enquote{lower} complexity of verifying $\leq$-minimal hypothesis can be explained by observing that a $\leq$-minimal hypothesis has size one (namely, $\{\alpha\}$ itself).

\begin{theorem}\label{thm:verification-el}
	For $\ELbot$, verification for $\subseteq$-minimal $\brave$-hypotheses is $\DP$-complete,
  whereas verification for $\ar$-hypotheses is $\DP$-hard and in $\PiP$.
\end{theorem}

\begin{proof}
	For membership, observe that $\Hyp$ is a $\subseteq$-minimal $\calS$-hypothesis if and only if (1) $\tup{\calT, \calA\cup \Hyp} \models_\calS \alpha$ and (2) for all subsets $\Hyp' \subsetneq \Hyp$, we have $\tup{\calT, \calA \cup \Hyp} \not\models_\calS \alpha$.
  In the case of $\brave$-hypotheses, (1) is instance checking for \ELbot and hence in \NP, while (2) can be checked in \coNP by universally guessing a subset $\calH'$ and repair $\calR$ of $\tup{\calT, \calA \cup \Hyp'}$ and checking that $\tup{\calT, \calR} \not\models \alpha$ in polynomial time.
	Hence, the problem is contained in \DP.
  Analogous reasoning under \ar semantics yields that (1) can be checked in $\co\NP$, whereas checking (2) requires an oracle to decide non-entailment under $\ar$ semantics for each $\calH'\subseteq \calH$.
  This shows $\co\NP^{\NP}$-membership.
	
	For hardness, we reduce from a combination of instance checking and its complement problem to our verification of hypotheses.
	Given an instance $\tup{\calK, \alpha_1, \alpha_2}$, the problem asks whether $\calK\models_\calS \alpha_1$ and $\calK\not\models_\calS \alpha_2$, where $\calS\in\{\brave,\ar\}$.
	This problem is $\DP$-complete because the first question is $\NP$-complete and the second question is $\co\NP$-complete under $\brave$ semantics and vice versa under $\ar$ semantics.
	For the reduction, assume $\alpha_1 = D(a)$, $\alpha_2= C(a)$, and $\calK = \tup{\calT,\calA}$.
  We construct a KB $\calK'$, an observation $\alpha$, and a hypothesis $\Hyp$ as illustrated next.
	Let $\calK' \coloneqq \tup{\calT',\calA}$ with $\calT' = \calT \cup \{C \subsum A, A\sqcap B\sqcap D \subsum Q\}$, $\alpha \coloneqq Q(a)$, and $\Hyp \coloneqq \{A(a), B(a)\}$ for fresh concepts $A,B,Q$.
	The instance is a valid abduction problem, since $\calK' \not\models_\calS \alpha$ (in particular, due to $B(a)$).
	Intuitively, $\Hyp$ is a Brave-hypothesis for $\alpha$ in $\calK'$ if and only if $\calK \models_\brave D(a)$ and $\Hyp$ is subset-minimal if and only if $\calK \not\models_\brave C(a)$.
\begin{shortonly}
  It remains to show correctness, i.e., $\Hyp$ is a $\subseteq$-minimal hypothesis for $\alpha$ in $\calK'$ if and only if $\calK \models_\brave \alpha_1$ and $\calK \not\models_\brave \alpha_2$.
\end{shortonly}
\begin{longonly}
	We next prove the correctness of reduction.
	\begin{claim}
		$\Hyp$ is a $\subseteq$-minimal hypothesis for $\alpha$ in $\calK'$ if and only if $\calK \models_\brave \alpha_1$ and $\calK \not\models_\brave \alpha_2$.
	\end{claim}
	\begin{claimproof}
  ``$\Longrightarrow$'': 
	Suppose $\Hyp$ is a $\subseteq$-minimal hypothesis for $Q(a)$ in $\calK'$.
	Observe that the only way to obtain the entailment $\calK' \models_\brave Q(a)$ is via the TBox axiom $A\sqcap B\sqcap D\subsum Q$, since no axiom in $\calK$ contains $Q$.
	Therefore, $\calK \models_\brave D(a)$ since otherwise, $\calK' \not\models_\brave D(a)$ and hence $\calK' \not\models_\brave Q(a)$.
	Moreover, we have $\calK \not\models_\brave C(a)$:
	Suppose to the contrary that $\calK \models_\brave C(a)$.
	Then $\calK' \models A(a)$ due to the axiom $C \sqsubseteq A$.
	Consequently, $\{B(a)\}$ is a \brave-hypothesis for $\alpha$ in $\calK'$, which is a contradiction to $\subseteq$-minimality of $\Hyp$.
	
	``$\Longleftarrow$'': Suppose $\calK \models D(a)$ and $\calK \not\models C(a)$.
	Then, $\tup{\calT', \calA\cup\Hyp} \models_\brave Q(a)$ since $\calK' \models D(a)$.
	Therefore $\Hyp$ is indeed a $\brave$-hypothesis for $\alpha$ in $\calK'$.
	For $\subseteq$-minimality, suppose to the contrary that there is a \brave-hypothesis $\Hyp' \subsetneq \Hyp$ for $\alpha$ in $\calK'$.
	Notice that $\Hyp'=\{B(a)\}$ since $B(a)$ can not be entailed from any other axiom in $\calK'$.
	However, this implies that $\calK' \models_\brave A(a)$, which can only be true if $\calK' \models C(a)$.
	As a result, we deduce that $\calK \models C(a)$, which is again a contradiction.
  \end{claimproof} 
\end{longonly}

	We conclude by observing that the above correctness proof works if we replace every $\brave$-entailment by $\ar$-entailment.
\end{proof}

We now turn to the case of \DLLite.
We begin by an observation on $\subseteq$-minimal (and $\leq$-minimal) \brave-hypotheses, namely that they always have cardinality $1$.

\begin{lemma}\label{lem:brave-hyp-size}
	For $\DLLite$, if $\calH$ is a $\subseteq$-minimal or $\leq$-minimal \brave-hypothesis for some \brave-abduction problem $\tup{\calK, \alpha}$, then $|\calH|=1$.
\end{lemma}
\begin{longonly}
\begin{proof}
  If \Hyp is $\leq$-minimal, this is obvious, as $\{\alpha\}$ is a \brave-hypothesis for $\tup{\calK, \alpha}$ (see \cref{sec:triv}).
  Now assume that \Hyp is $\subseteq$-minimal and let $\calK = \tup{\calT, \calA}$.
	This means that we have $\tup{\calT, \calA \cup \Hyp} \models_\brave \alpha$.
	Hence, there is a repair $\calR \subseteq \calA \cup \Hyp$ such that $\tup{\calT, \calR} \models \alpha$.
	In other words, \calR is a \calT-support of $\alpha$, that is, an ABox consistent with \calT that entails $\alpha$ in \calT.
  As subset-minimal supports with respect to \DLLite TBoxes are always of size $1$, there is an assertion $\beta \in \calR$ such that $\tup{\calT, \{\beta\}} \models \alpha$.
	This implies that also $\tup{\calT, \calA \cup \{\beta\}} \models_\brave \alpha$, so $\{\beta\}$ is a \brave-hypothesis for $\alpha$ in \calK.
	As \Hyp is $\subseteq$-minimal and $\beta \in \Hyp$, this implies that $\Hyp = \{\beta\}$.
\end{proof}
\end{longonly}

The next theorem establishes the complexity of the verification problem for \brave-hypotheses in \DLLite, in the general, $\leq$-minimal and $\subseteq$-minimal case.

\begin{theorem}\label{thm:verification-dllite-brave}
	For $\DLLite$, the verification problem for general, $\leq$-minimal and $\subseteq$-minimal \brave-hypotheses is \NL-complete.
\end{theorem}
\begin{longonly}
\begin{proof}
	We first show membership.
	For general hypotheses, this can be shown 
	By \cref{lem:brave-hyp-size}, a given ABox is a $\subseteq$-minimal hypothesis if and only if it is a $\leq$-minimal hypothesis if and only if it is a hypothesis and has size $1$.
  Hence, we simply have to additionally check whether $|\Hyp| = 1$ in the algorithm for general hypotheses.

	
	Hardness for all three kinds of hypotheses can be shown by a straightforward reduction from reachability in directed graphs, using almost the same construction as for hardness in \cref{thm:signature-complexity-dllite}.
	For a directed graph $G = (V,E)$ $s,t \in V$, we can construct $\calK = \tup{\calT, \calA}$ exactly the same as in that proof in logarithmic space.
  Now instead of asking for existence of a hypothesis over signature $\{A_s, a\}$, we ask whether the ABox $\{A_s(a)\}$ is a general, $\leq$-minimal, or $\subseteq$-minimal \brave-hypothesis for $A_t(a)$.
  As before, we can observe that there is an $s$-$t$ path in $G$ if and only if $A_s(a)$ is a (general) \brave-hypothesis for $A_t(a)$ in \calK, and since $\{A_s(a)\}$ is a singleton set, the same is true with respect to $\leq$-minimal and $\subseteq$-minimal hypotheses.
\end{proof}
\end{longonly}

Finally, we turn towards the case of \ar semantics.

\begin{theorem}\label{thm:verification-dllite-ar}
	For \DLLite, the verification problem for \ar-hypotheses is 
	(1) $\co\NP$-complete for general and $\leq$-minimal hypotheses, and
	(2) $\DP$-hard for $\subseteq$-minimal ones with membership in $\PiP$.
\end{theorem}

\begin{proof}
	\emph{General hypotheses:}
	Regarding membership, observe that the question can be answered by checking whether $\tup{\calT,\calA\cup\calH}\models_\ar \alpha$.
	Hence, the complexity follows from that of $\ar$-entailment for $\DLLite$.
	For hardness, we reuse the following reduction from unsatisfiability and AR-entailment~\cite{bienvenu2019computing}.
	Let $\varphi = \{c_1, \dots, c_k\}$ over propositions $X=\{x_1,\dots, x_n\}$, where the $c_i$ are clauses.
  We construct $\calK = \tup{\calT, \calA}$ using a single concept name $A$ and role names $N = \{P, N, U\} $, where
	\begin{align*}
	\calT &= \{\exists P^{-} \subsum \neg \exists N^{-}, \exists P \subsum \neg \exists U^{-}, \exists N \subsum \neg \exists U^{-}, \exists U \subsum A \}, \text{ and} \\
	\calA &= \{P(c_j,x_i) \mid x_i \in c_j \} \cup \{N(c_j,x_i) \mid \neg x_i \in c_j \} \cup \{U(a,c_j) \mid j \leq k \}.
\end{align*}
	Moreover, let $\alpha \dfn A(a)$.
	It is known that $\calK \models_\ar A(a)$ if and only if $\varphi$ is unsatisfiable~\cite{bienvenu2019computing}.
	To show hardness of the verification problem at hand, let $\Hyp \dfn \{U(a,c_j) \mid j \leq k\}$ and $\calK' \dfn \tup{\calT, \calA \setminus \Hyp}$.
	Clearly, $\Hyp$ is an $\ar$-hypothesis for $\alpha$ in $\calK'$ if and only if $\tup{\calT, \calA \cup \Hyp} \models_\ar \alpha$ if and only if $\varphi$ is unsatisfiable.

	\emph{Cardinality-minimal hypotheses:}
  For membership, recall that for a given \ar-abduction problem $\tup{\calK, \alpha}$, the singleton set $\{\alpha\}$ is a solution if and only if there is any solution by \cref{lem:conf-pres}. 
  Hence, we can use the algorithm for general hypotheses and additionally check that $|\Hyp| = 1$, yielding \coNP-membership.

  For hardness, we again adapt the reduction from unsatisfiability to \ar-entailment. 
  In particular, we modify the given CNF-formula before applying the reduction to ensure that a specific singleton ABox is an \ar-hypothesis if and only if the CNF-formula is unsatisfiable.
  Let $\varphi = \{c_1, \dots, c_k\}$ over variables $X = \{x_1, \dots, x_n\}$.
  Define
  \begin{align*}
    c_i' &\coloneqq c_i \cup \{x_{n+1}\} \text{ for } 1 \leq i \leq k, \\
    c_{k+1}' &\coloneqq \neg x_{n+1} \lor x_{n+2}, \text{ and} \\
    c_{k+2}' &\coloneqq \neg x_{n+2}
  \end{align*}
  and let $\varphi_1 \coloneqq \{c_1', \dots, c_{k+1}'\}$ and $\varphi_2 \coloneqq \varphi_1 \cup \{c_{k+2}'\}$.
  Analogously to the construction of \calK from $\varphi$ in the hardness proof for general hypotheses above, we construct knowledge bases $\calK_i = \tup{\calT,\calA_i}$ from $\varphi_i$ for $i \in \{1,2\}$.
  Further, define $\calK_2' \coloneqq \tup{\calT, \calA_2 \setminus \{U(a, c_{k+2})\}}$.
\begin{shortonly}
  In order to show \coNP-hardness, it remains to show that $\tup{\calK_2', A(a)}$ is a valid \ar-abduction problem and $\Hyp = \{U(a, c_{k+2})\}$ is a ($\leq$-minimal) solution to it if and only if $\varphi$ is unsatisfiable.
\end{shortonly}
\begin{longonly}
  The following claim now establishes $\coNP$-hardness. 
  \begin{claim}
    $\tup{\calK_2', A(a)}$ is a valid \ar-abduction problem and $\Hyp = \{U(a, c_{k+2})\}$ is a ($\leq$-minimal) solution to it if and only if $\varphi$ is unsatisfiable.
  \end{claim}
  \begin{claimproof}
  Obviously, every satisfying asignment of $\varphi_2$ assigns $x_{n+2}$ and $x_{n+1}$ to $0$, and hence $\varphi_2$ and $\varphi$ are equisatisfiable.
  The formula $\varphi_1$, however, is always satisfiable, as we can simply assign both $x_{n+1}$ and $x_{n+2}$ to $1$.
  Now note that $\varphi$ is unsatisfiable if and only if $\calK_2 \models_\ar A(a)$, because of correctness of the original reduction and equisatisfiability of $\varphi$ and $\varphi_2$. 
  Next observe that entailment of $A(a)$ in $\calK_2'$ only depends on the assertions corresponding to clauses $c_1, \dots, c_{k+1}$, as $U(a, c_{k+2})$ is not present.
  Hence, despite keeping the assertion $N(c_{k+2}, x_{n+2})$, we have that $\calK_1 \models_\ar A(a)$ if and only if $\calK_2' \models_\ar A(a)$.
  Because of correctness of the original reduction and the fact that $\varphi_1$ is satisfiable, this means that $\calK_2' \not\models_\ar A(a)$.
  Combining these two observations, it follows that $\tup{\calK_2', A(a)}$ is a valid \ar-abduction problem, and $\calH \coloneqq \{U(a, c_{k+2})\}$ is a solution to it if and only if $\varphi$ is unsatisfiable.
  Furthermore, if $\calH$ is a solution, then it is also $\leq$-minimal as $|\calH| = 1$.
  \end{claimproof}
\end{longonly}
	
	\emph{Subset-minimal hypotheses:}
	We can prove \PiP-membership similar to the case of \ELbot in Theorem~\ref{thm:verification-el}.
	For $\DP$-hardness, we reuse the above reduction but first 
	introduce some terminology.
	Given a formula $\varphi$ in CNF, a collection $\psi\subseteq \varphi$ of clauses is a \emph{minimal unsatisfiable subset} (MUS) of $\varphi$ if $\psi$ is unsatisfiable but $\psi'$ is satisfiable for every $\psi' \subset\psi$.
	It can be observed that the subset-minimal \ar-hypotheses $\calH$ for $\alpha$ in $\calK'$ correspond precisely to MUSes $\psi_\Hyp$ for $\varphi$ by taking $c_j\in \psi_\calH\iff U(a,c_j)\in \calH$.
	Then, the claim follows by observing that the problem to decide if a set of clauses is a MUS is $\DP$-hard~\cite{Liberatore05}.
	For hardness, we reuse the reduction from above and encode a given set $\psi$ into the hypothesis as $\Hyp_\psi= \{U(a,c_j)\mid c_j\in\psi\}$.	
\end{proof}


\section{Conclusion and Future Work} 

\textbf{Summary.} In this paper, we provided an initial study on ABox abduction under repair semantics building on the work from \cite{du2015towards}.
Our main contributions include new minimality criteria for preferred hypotheses w.r.t.\  inconsistent KBs and initial complexity results for the existence and the verification problem for flat ABox abduction and atomic BIQs as observations. 
Our results on combined complexity show that with an unrestricted signature, the complexity can be lower than for the entailment under repair semantics, while signature restrictions can make the problems computationally harder.
Verification stays as hard as deciding classical entailment, but the choice of minimality criteria can increase the complexity (e.g., $\subseteq$-minimality).

\medskip
\noindent
\textbf{Future Work.} For our initial setting considered, we have a complete picture of the complexity regarding \brave semantics, whereas the complexity analysis for  \ar semantics  has some gaps.
It seems that  these gaps can  be explained by the non-convex behavior 
of $\ar$-hypotheses that was illustrated in Section~\ref{sec:abd}. We plan to explore these effects further and complete the complexity landscape for the considered problems and more expressive formulas as observations.
Moreover, the complexity when considering conflict-confining hypotheses also remains open for certain cases, even for \brave-semantics.
Having established a complete picture regarding the combined complexity, we also intend to see the effect of a fixed TBox by considering the {data} complexity. 
One can observe that several results from the current paper already transfer to the data complexity since the employed reductions result in a fixed TBox.

There are many directions for future work regarding extensions of the fairly limited initial setting studied here. 
One particularly interesting direction is to explore the related problems from the literature on abduction, such as \emph{necessity} and \emph{relevance} of axioms in hypotheses, which have been treated to a certain extent in~\cite{bienvenu2019computing}.
Moreover, the abduction problem with size restrictions has been considered before in propositional logic~\cite{MahmoodMS21,Mahmood22}.
In our setting, it seems interesting to impose size restrictions for a hypothesis but also for the corresponding set of conflicts.
Additionally, the signature-based settings considered previously only restrict concepts and roles~\cite{Koopmann21a}.
This has the effect that the hypotheses may get exponentially large already for $\ELbot$.
It is therefore worth exploring whether the inconsistency of KBs poses any additional challenges resulting in another blow-up.
We also aim to define a suitable and meaningful notion of semantically minimal hypothesis under repair semantics in future work.

\bibliography{../main}

\begin{thebibliography}{10}

\bibitem{AlrabbaaBFHKKKK24}
Christian Alrabbaa, Stefan Borgwardt, Tom Friese, Anke Hirsch, Nina Knieriemen,
  Patrick Koopmann, Alisa Kovtunova, Antonio Kr{\"{u}}ger, Alexej Popovic, and
  Ida S.~R. Siahaan.
\newblock Explaining reasoning results for {OWL} ontologies with evee.
\newblock In {\em Proceedings of the 21st International Conference on
  Principles of Knowledge Representation and Reasoning, {KR} 2024}. {AAAI}
  Press, 2024.
\newblock URL: \url{https://doi.org/10.24963/kr.2024/67}.

\bibitem{DBLP:books/daglib/0041477}
Franz Baader, Ian Horrocks, Carsten Lutz, and Ulrike Sattler.
\newblock {\em An Introduction to Description Logic}.
\newblock Cambridge University Press, 2017.

\bibitem{bienvenu2008complexity}
Meghyn Bienvenu.
\newblock Complexity of abduction in the {$\mathcal{EL}$} family of lightweight
  description logics.
\newblock In {\em Proceedings of the Eleventh International Conference on
  Principles of Knowledge Representation and Reasoning, {KR} 2008}, pages
  220--230. AAAI Press, 2008.
\newblock URL: \url{http://www.aaai.org/Library/KR/2008/kr08-022.php}.

\bibitem{Bienvenu-KIJ-20}
Meghyn Bienvenu.
\newblock A short survey on inconsistency handling in ontology-mediated query
  answering.
\newblock {\em K{\"{u}}nstliche Intell.}, 34(4):443--451, 2020.
\newblock \href {https://doi.org/10.1007/S13218-020-00680-9}
  {\path{doi:10.1007/S13218-020-00680-9}}.

\bibitem{BienvenuB16}
Meghyn Bienvenu and Camille Bourgaux.
\newblock Inconsistency-tolerant querying of description logic knowledge bases.
\newblock In {\em Reasoning Web: Logical Foundation of Knowledge Graph
  Construction and Query Answering - 12th International Summer School 2016},
  volume 9885 of {\em Lecture Notes in Computer Science}, pages 156--202.
  Springer, 2016.
\newblock \href {https://doi.org/10.1007/978-3-319-49493-7\_5}
  {\path{doi:10.1007/978-3-319-49493-7\_5}}.

\bibitem{bienvenu2019computing}
Meghyn Bienvenu, Camille Bourgaux, and Fran{\c{c}}ois Goasdou{\'e}.
\newblock Computing and explaining query answers over inconsistent {DL-Lite}
  knowledge bases.
\newblock {\em Journal of Artificial Intelligence Research}, 64:563--644, 2019.

\bibitem{BienvenuR13}
Meghyn Bienvenu and Riccardo Rosati.
\newblock Tractable approximations of consistent query answering for robust
  ontology-based data access.
\newblock In Francesca Rossi, editor, {\em {IJCAI} 2013, Proceedings of the
  23rd International Joint Conference on Artificial Intelligence, Beijing,
  China, August 3-9, 2013}, pages 775--781. {IJCAI/AAAI}, 2013.
\newblock URL:
  \url{http://www.aaai.org/ocs/index.php/IJCAI/IJCAI13/paper/view/6904}.

\bibitem{BoborovaKHP24}
Janka Boborov{\'{a}}, Jakub Kloc, Martin Homola, and J{\'{u}}lia
  Pukancov{\'{a}}.
\newblock Towards {CATS:} {A} modular {ABox} abduction solver based on
  black-box architecture (extended abstract).
\newblock In {\em Proceedings of the 37th International Workshop on Description
  Logics {(DL} 2024)}, volume 3739 of {\em {CEUR} Workshop Proceedings}.
  CEUR-WS.org, 2024.
\newblock URL: \url{https://ceur-ws.org/Vol-3739/abstract-8.pdf}.

\bibitem{BFPS-AIJ-15}
Piero~A. Bonatti, Marco Faella, Iliana~M. Petrova, and Luigi Sauro.
\newblock A new semantics for overriding in description logics.
\newblock {\em Artif. Intell.}, 222:1--48, 2015.
\newblock \href {https://doi.org/10.1016/J.ARTINT.2014.12.010}
  {\path{doi:10.1016/J.ARTINT.2014.12.010}}.

\bibitem{Calvanese2013}
Diego Calvanese, Magdalena Ortiz, Mantas Simkus, and Giorgio Stefanoni.
\newblock Reasoning about explanations for negative query answers in {DL-Lite}.
\newblock {\em J. Artif. Intell. Res.}, 48:635--669, 2013.
\newblock \href {https://doi.org/10.1613/jair.3870}
  {\path{doi:10.1613/jair.3870}}.

\bibitem{CMMSV-ISWC-15}
Giovanni Casini, Thomas~Andreas Meyer, Kodylan Moodley, Uli Sattler, and Ivan
  Varzinczak.
\newblock Introducing defeasibility into {OWL} ontologies.
\newblock In {\em Proceedings of the 14th International Semantic Web
  Conference}, volume 9367 of {\em Lecture Notes in Computer Science}, pages
  409--426. Springer, 2015.
\newblock \href {https://doi.org/10.1007/978-3-319-25010-6\_27}
  {\path{doi:10.1007/978-3-319-25010-6\_27}}.

\bibitem{Ceylan2020}
{\.I}smail~{\.I}lkan Ceylan, Thomas Lukasiewicz, Enrico Malizia, Cristian
  Molinaro, and Andrius Vaicenavicius.
\newblock Explanations for negative query answers under existential rules.
\newblock In {\em Proceedings of the 17th International Conference on
  Principles of Knowledge Representation and Reasoning, {KR} 2020}, pages
  223--232. {AAAI} Press, 2020.
\newblock URL: \url{https://doi.org/10.24963/kr.2020/23}.

\bibitem{Del-PintoS19}
Warren Del{-}Pinto and Renate~A. Schmidt.
\newblock {ABox} abduction via forgetting in {$\mathcal{ALC}$}.
\newblock In {\em Proceedings of the Thirty-Third {AAAI} Conference on
  Artificial Intelligence}, pages 2768--2775. AAAI Press, 2019.
\newblock URL: \url{https://doi.org/10.1609/aaai.v33i01.33012768}.

\bibitem{du2017practical}
Jianfeng Du, Hai Wan, and Huaguan Ma.
\newblock Practical {TBox} abduction based on justification patterns.
\newblock In {\em Proceedings of the Thirty-First {AAAI} Conference on
  Artificial Intelligence}, pages 1100--1106, 2017.
\newblock URL:
  \url{http://aaai.org/ocs/index.php/AAAI/AAAI17/paper/view/14402}.

\bibitem{du2015towards}
Jianfeng Du, Kewen Wang, and Yi{-}Dong Shen.
\newblock Towards tractable and practical abox abduction over inconsistent
  description logic ontologies.
\newblock In {\em Proceedings of the Twenty-Ninth {AAAI} Conference on
  Artificial Intelligence}, pages 1489--1495. {AAAI} Press, 2015.
\newblock URL: \url{https://doi.org/10.1609/aaai.v29i1.9393}.

\bibitem{Elsenbroich2006}
Corinna Elsenbroich, Oliver Kutz, and Ulrike Sattler.
\newblock A case for abductive reasoning over ontologies.
\newblock In {\em Proceedings of the OWLED'06 Workshop on {OWL:} Experiences
  and Directions}, 2006.
\newblock URL: \url{http://ceur-ws.org/Vol-216/submission\_25.pdf}.

\bibitem{GlimmKW22}
Birte Glimm, Yevgeny Kazakov, and Michael Welt.
\newblock Concept abduction for description logics.
\newblock In {\em Proceedings of the 35th International Workshop on Description
  Logics {(DL} 2022) co-located with Federated Logic Conference (FLoC 2022),
  Haifa, Israel, August 7th to 10th, 2022}, 2022.
\newblock URL: \url{https://ceur-ws.org/Vol-3263/paper-11.pdf}.

\bibitem{Haifani2022}
Fajar Haifani, Patrick Koopmann, Sophie Tourret, and Christoph Weidenbach.
\newblock Connection-minimal abduction in {$\mathcal{EL}$} via translation to
  {FOL}.
\newblock In Jasmin Blanchette, Laura Kov{\'{a}}cs, and Dirk Pattinson,
  editors, {\em Proceedings of the 11th International Joint Conference on
  Automated Reasoning {IJCAR} 2022}, volume 13385 of {\em Lecture Notes in
  Computer Science}, pages 188--207. Springer, 2022.
\newblock URL: \url{https://doi.org/10.1007/978-3-031-10769-6\_12}.

\bibitem{HallandBritzABox2012}
Ken Halland and Katarina Britz.
\newblock {ABox} abduction in $\mathcal{ALC}$ using a {DL} tableau.
\newblock In {\em Proceedings of the South African Institute for Computer
  Scientists and Information Technologists Conference}, SAICSIT '12, page
  51–58, New York, NY, USA, 2012. Association for Computing Machinery.
\newblock \href {https://doi.org/10.1145/2389836.2389843}
  {\path{doi:10.1145/2389836.2389843}}.

\bibitem{Homola2023}
Martin Homola, J{\'{u}}lia Pukancov{\'{a}}, Janka Boborov{\'{a}}, and Iveta
  Balintov{\'{a}}.
\newblock Merge, explain, iterate: A combination of {MHS} and {MXP} in an
  {ABox} abduction solver.
\newblock In {\em Proceedings of the 18th European Conference on Logics in
  Artificial Intelligence, {JELIA} 2023}, volume 14281 of {\em Lecture Notes in
  Computer Science}, pages 338--352. Springer, 2023.
\newblock URL: \url{https://doi.org/10.1007/978-3-031-43619-2\_24}.

\bibitem{Klarman2011}
Szymon Klarman, Ulle Endriss, and Stefan Schlobach.
\newblock {ABox} abduction in the description logic $\mathcal{ALC}$.
\newblock {\em Journal of Automated Reasoning}, 46(1):43--80, 2011.

\bibitem{Koopmann21a}
Patrick Koopmann.
\newblock Signature-based abduction with fresh individuals and complex concepts
  for description logics.
\newblock In {\em Proceedings of the Thirtieth International Joint Conference
  on Artificial Intelligence}, pages 1929--1935. ijcai.org, 2021.
\newblock URL: \url{https://doi.org/10.24963/ijcai.2021/266}.

\bibitem{Koopmann2020}
Patrick Koopmann, Warren Del{-}Pinto, Sophie Tourret, and Renate~A. Schmidt.
\newblock Signature-based abduction for expressive description logics.
\newblock In {\em Proceedings of the 17th International Conference on
  Principles of Knowledge Representation and Reasoning, {KR} 2020}, pages
  592--602. {AAAI} Press, 2020.
\newblock URL: \url{https://doi.org/10.24963/kr.2020/59}.

\bibitem{LLRRS-JWS-15}
Domenico Lembo, Maurizio Lenzerini, Riccardo Rosati, Marco Ruzzi, and
  Domenico~Fabio Savo.
\newblock Inconsistency-tolerant query answering in ontology-based data access.
\newblock {\em J. Web Semant.}, 33:3--29, 2015.
\newblock \href {https://doi.org/10.1016/J.WEBSEM.2015.04.002}
  {\path{doi:10.1016/J.WEBSEM.2015.04.002}}.

\bibitem{Liberatore05}
Paolo Liberatore.
\newblock Redundancy in logic {I:} {CNF} propositional formulae.
\newblock {\em Artif. Intell.}, 163(2):203--232, 2005.
\newblock URL: \url{https://doi.org/10.1016/j.artint.2004.11.002}, \href
  {https://doi.org/10.1016/J.ARTINT.2004.11.002}
  {\path{doi:10.1016/J.ARTINT.2004.11.002}}.

\bibitem{lukasiewicz2022explanations}
Thomas Lukasiewicz, Enrico Malizia, and Cristian Molinaro.
\newblock Explanations for negative query answers under inconsistency-tolerant
  semantics.
\newblock In {\em Proceedings of the Thirty-First International Joint
  Conference on Artificial Intelligence, {IJCAI-22}}, pages 2705--2711.
  ijcai.org, 2022.
\newblock \href {https://doi.org/10.24963/IJCAI.2022/375}
  {\path{doi:10.24963/IJCAI.2022/375}}.

\bibitem{Mahmood22}
Yasir Mahmood.
\newblock {\em Parameterized aspects of team-based formalisms and logical
  inference}.
\newblock PhD thesis, University of Hanover, Germany, 2022.
\newblock URL: \url{https://www.repo.uni-hannover.de/handle/123456789/13169},
  \href {https://doi.org/10.15488/13064} {\path{doi:10.15488/13064}}.

\bibitem{MahmoodMS21}
Yasir Mahmood, Arne Meier, and Johannes Schmidt.
\newblock Parameterized complexity of abduction in schaefer's framework.
\newblock {\em J. Log. Comput.}, 31(1):266--296, 2021.
\newblock URL: \url{https://doi.org/10.1093/logcom/exaa079}, \href
  {https://doi.org/10.1093/LOGCOM/EXAA079} {\path{doi:10.1093/LOGCOM/EXAA079}}.

\bibitem{PT-JAR-18}
Maximilian Pensel and Anni{-}Yasmin Turhan.
\newblock Reasoning in the defeasible description logic $\mathcal{EL}$ -
  computing standard inferences under rational and relevant semantics.
\newblock {\em Int. J. Approx. Reason.}, 103:28--70, 2018.
\newblock \href {https://doi.org/10.1016/J.IJAR.2018.08.005}
  {\path{doi:10.1016/J.IJAR.2018.08.005}}.

\bibitem{DBLP:books/daglib/0092426}
Nicholas Pippenger.
\newblock {\em Theories of computability}.
\newblock Cambridge University Press, 1997.

\bibitem{WeiKleinerDragisicLambrix2014}
Fang Wei{-}Kleiner, Zlatan Dragisic, and Patrick Lambrix.
\newblock Abduction framework for repairing incomplete $\mathcal{EL}$
  ontologies: Complexity results and algorithms.
\newblock In {\em Proceedings of the Twenty-Eighth {AAAI} Conference on
  Artificial Intelligence}, pages 1120--1127. AAAI Press, 2014.
\newblock URL:
  \url{http://www.aaai.org/ocs/index.php/AAAI/AAAI14/paper/view/8239}.

\end{thebibliography}

\end{document}